\documentclass[letterpaper,11pt]{article}  
\usepackage{amsfonts}
\usepackage[pdftex]{graphicx}
\usepackage{amsmath,amssymb,amsthm}
\usepackage{natbib}
\usepackage{hyperref}
\usepackage{bm}
\usepackage{caption}
\usepackage{subcaption}
\usepackage{adjustbox,lipsum}
\usepackage{fancyhdr}
\usepackage{sectsty}
\usepackage{stmaryrd}
\usepackage{setspace}                      
\usepackage{booktabs}    
\usepackage{pdfsync}        
\usepackage[backgroundcolor=blue!40,linecolor=blue!40]{todonotes}
\usepackage{fancybox}
\usepackage{appendix}
\usepackage{placeins}
\usepackage{gensymb}

\definecolor{cornellred}{RGB}{179,27,27} 
\definecolor{cornellblue}{RGB}{00,00,170}
\definecolor{cornellgrey}{RGB}{96,94,92}

\newtheoremstyle{myplain}
  {9pt}
  {9pt}
  {\itshape}
  {\parindent}
  {\scshape}
  {:}
  {.5em}
  {}
\newtheoremstyle{mydefinition}
  {9pt}
  {9pt}
  {\itshape}
  {\parindent}
  {\scshape}
  {:}
  {.5em}
  {}
\newtheoremstyle{myremark}
  {9pt}
  {9pt}
  {}
  {\parindent}
  {\scshape}
  {:}
  {.5em}
  {}
\theoremstyle{myplain}
\newtheorem{theorem}{Theorem}

\theoremstyle{mydefinition}

\newtheorem{assumption}{Assumption}

\newtheorem{definition}{Definition}

\theoremstyle{myremark}

\newtheorem{remark}{Remark}

\renewcommand{\cite}{\citet}

\usetikzlibrary{calc}
\def\centerarc[#1](#2)(#3:#4:#5){ \draw[#1] ($(#2)+({#5*cos(#3)},{#5*sin(#3)})$) arc (#3:#4:#5);}

\hypersetup{colorlinks=true, linkcolor=blue, citecolor=blue}                          
\numberwithin{equation}{section}

\begin{document}

\title{A Simple, Short, but Never-Empty Confidence Interval \\ for Partially Identified Parameters}
\date{\today}
\author{J\"{o}rg Stoye\thanks{Department of Economics, 
Cornell University, stoye@cornell.edu. Thanks to Johannes Haushofer, Jonathan de Quidt, and Chris Roth for an inquiry that motivated this work and for sharing and explaining their data. Financial support through NSF Grant SES-1824375 is gratefully acknowledged.}}

\maketitle
\begin{abstract}
This paper revisits the simple, but empirically salient, problem of inference on a real-valued parameter that is partially identified through upper and lower bounds with asymptotically normal estimators. A simple confidence interval is proposed and is shown to have the following properties:
\begin{itemize}
\item It is never empty or awkwardly short, including when the sample analog of the identified set is empty.
\item It is valid for a well-defined pseudotrue parameter whether or not the model is well-specified.
\item It involves no tuning parameters and minimal computation.
\end{itemize}
Computing the interval requires concentrating out one scalar nuisance parameter. In most cases, the practical result will be simple: To achieve $95\%$ coverage, report the union of a simple $90\%$ (!) confidence interval for the identified set and a standard $95\%$ confidence interval for the pseudotrue parameter.

For uncorrelated estimators --notably if bounds are estimated from distinct subsamples-- and conventional coverage levels, validity of this simple procedure can be shown analytically. The case obtains in the motivating empirical application \citep{Haushofer}, in which improvement over existing inference methods is demonstrated. More generally, simulations suggest that the novel confidence interval has excellent length and size control. This is partly because, in anticipation of never being empty, the interval can be made \textit{shorter} than conventional ones in relevant regions of sample space.
\end{abstract}
\vfill

\pagebreak
\onehalfspacing

\section{Introduction}

Inference under partial identification is by now the subject of a broad literature.\footnote{See \citet{Manski2003} for an early monograph, \citet{Tamer10} for a historical introductions, and \citet{CS17} and \citet{MolinariHOE} for recent surveys that extensively cover inference.} Only recently did attention turn to the following concern: If a partially identified model is  misspecified, this may manifest in either an empty or --and arguably worse-- in a misleadingly small confidence region. That is, misspecified inference can be spuriously precise.

The reason is that most confidence regions used in partial identification invert tests of $H_0:\theta \in \Theta_I$; here, $\theta$ is a parameter and $\Theta_I$ is the identified set. If $H_0$ is rejected at every $\theta$, the confidence region is empty. If $H_0$ is barely not rejected at a few parameter values, the confidence region may be very small. This issue is empirically relevant. For example, an empty sample analog of $\Theta_I$ occurs in \citet{Haushofer}, whose inquiry sparked the present research and whose data are reanalyzed below.

The literature on this issue is still young. \cite{PT11} provide an early diagnosis. \citet{KW13} propose a notion of pseudotrue identified set and an estimator thereof. \citet{MolinariHOE} explains the issue in detail and highlights it as important area for further investigation. The most thorough treatment is by \citet{AndrewsKwon19}, who emphasize the issue's importance and provide a general inference method that avoids spurious precision and ensures coverage of a pseudotrue identified set.

The present paper is in the spirit of \citet{AndrewsKwon19}. I focus on the simple but empirically salient case of a scalar parameter with upper and lower bounds whose estimators are jointly asymptotically normal. That is, I revisit the setting of \citet[][without their superefficiency assumption]{IM04} and \citet{Stoye09}. For this setting, I propose a confidence interval with the following features:
\begin{itemize}
\item It is never empty nor very short (a lower bound on its length is reported later).
\item It exhibits asymptotically guaranteed coverage uniformly over the identified set and additionally for a well-defined pseudotrue parameter.
\item It tends to be \textit{shorter} than more conventional intervals in benign cases, including in the empirical application.
\item It is free of tuning parameters and trivial to compute.
\end{itemize}
For target coverage of $95\%$ and for the special case of uncorrelated estimators, e.g. in this paper's empirical application, the confidence interval can be verbally defined as follows:
\begin{itemize}
\item Add $\pm 1.64$ standard errors to estimators of upper and lower bounds.
\item Also compute an average of the estimators that is weighted by their standard errors, as well as the corresponding standard error. Add $\pm1.96$ of those standard errors to the average.
\item Report the union of the intervals.
\end{itemize}
While this paper generally proposes a somewhat less ``cute" procedure with broader applicability, this specialized finding is probably the most striking part.\footnote{Full disclaimer: I discovered it by simulation and initially assumed a bug.} Neither of the above two intervals is valid by itself; it is just that their coverage events are correlated in exactly the right way.

Section \ref{sec:background} develops the proposal more formally and gives an intuition for why it works, though proofs are relegated to the Appendix. Section \ref{sec:numerical} provides a numerical illustration and Section \ref{sec:empirical} an application to the data that motivated this research. Section \ref{sec:conclusion} concludes.

\section{A Misspecification-Adaptive Confidence Interval}\label{sec:background}

While the interpretation of what follows is inference on a scalar parameter $\theta$, the only assumption is that one has well-behaved estimators of two other parameter values.
\begin{assumption}\label{as:1}
There exist estimators $(\hat{\theta}_L,\hat{\theta}_U)$ with probability limits $(\theta_L,\theta_U) \in \bm{R}^2$ such that
\begin{equation*}
\sqrt{n}\left( 
\begin{array}{c}
\hat{\theta}_L-\theta_L \\ 
\hat{\theta}_U-\theta_U
\end{array}
\right)
\overset{d}{\to}
N \left( \left(
\begin{array}{c}
0 \\ 
0
\end{array} \right),
\left(
\begin{array}{cc}
\sigma_L^2 & \rho \sigma_L \sigma_U \\ 
\rho \sigma_L \sigma_U & \sigma_U^2
\end{array} \right)
\right),
\end{equation*}
where $\sigma_L,\sigma_U >0$ and consistent estimators $(\hat{\sigma}_L,\hat{\sigma}_U,\hat{\rho})\overset{p}{\to}(\sigma_L,\sigma_U,\rho)$ are available.
\end{assumption}
The motivation is that the researcher estimates an identified set $\Theta_I \equiv [\theta_L,\theta_U]$ containing a true parameter value $\theta$. Assumption \ref{as:1} is unrestrictive if, as in the empirical application, $(\hat{\theta}_L,\hat{\theta}_U)$ are smooth functions of sample moments. It is unlikely to hold for intersection bounds \citep{AS13,CLR13} and will hold for bounds that result from projecting a higher-dimensional identified set \citep{BCS17,KMS19}, including components of partially identified vectors, only in benign cases.

The obvious estimator of $\Theta_I$ is $[\hat{\theta}_L,\hat{\theta}_U]$, but defining a confidence interval is delicate. Following \citet{IM04}, the literature mostly focuses on confidence intervals that (asymptotically) contain the true parameter value with prespecified probability $(1-\alpha)$, irrespective of its location in $\Theta_I$, i.e. confidence intervals that control $\inf_{\theta \in \Theta_I} \Pr(\theta \in CI)$. Finding such intervals is subtle because the nature of the testing problem qualitatively depends on the length $\Delta \equiv \theta_U-\theta_L$ of $\Theta_I$. Heuristically, this problem is one-sided if $\Delta$ is ``large" and two-sided if it is ``short," i.e. near point identification. Ascertaining which case obtains is subject to difficulties reminiscient of post-model selection inference \citep{LP05} and parameter-on-the-boundary issues \citep{Andrews00}.

The literature on how to circumvent this issue is by now considerable. Most approaches invert a test, that is, they report all values of $\theta$ for which $H_0:\theta \in \Theta_I$ was not rejected. Any such confidence set can be empty; in this paper's settings, that will happen if $\hat{\theta}_L$ is much larger than $\hat{\theta}_U$, where the meaning of ``much" varies across papers. This feature can be advertised as an embedded specification test but may not be wanted.\footnote{That was the sales pitch in \citet{Stoye09}, but not all referees were sold on it. The embedded specification test is analyzed in more detail by \citet{AS10}.} Arguably even more problematic is that, if the model is misspecified, a test inversion confidence interval can be short, suggesting precision when the true issue is misspecification. A specification test will not resolve this: In this paper's setting, the best-practice such test \citep{BCS15} just reports whether the test inversion interval is empty.\footnote{This equivalence does not generalize, but \citet{AndrewsKwon19} show that in ``slightly misspecified" parameter regimes, spuriously precise inference generally coexists with low power of specification tests.}  

Addressing this concern requires a notion of coverage for the case of misspecification, i.e. if $\theta_L>\theta_U$. Following \citet{AndrewsKwon19}, define the pseudotrue identified set
\begin{eqnarray*}
\Theta_I^* &\equiv &\Theta_I \cup \{\theta^*\} \\
\theta^* &\equiv &\frac{\sigma_U \theta_L + \sigma_L \theta_U}{\sigma_L+\sigma_U}.
\end{eqnarray*}
This definition is natural because $\Theta_I^*=\arg\min_\theta \max\{(\theta-\theta_U)/\sigma_U,(\theta_L-\theta)/\sigma_L,0\}$; thus, $\Theta_I^*$ is the estimand implied by the frequent choice of $\max\{(\theta-\hat{\theta}_U)/\hat{\sigma}_U,(\hat{\theta}_L-\theta)/\hat{\sigma}_L,0\}$ as test statistic. Note also that $\Theta_I^*$ is never empty and that $\Theta_I^*=\Theta_I$ whenever $\Theta_I \neq \emptyset$.

The revised notion of validity of a confidence interval is as follows:
\begin{definition} \label{def:valid}
A confidence interval CI has asymptotic coverage of $(1-\alpha)$ if
$$\lim_{n \to \infty}  \inf_{\theta \in \Theta_I^*} \Pr(\theta \in CI) \geq 1-\alpha.$$
\end{definition}
Forcing coverage of $\theta^*$ will ensure that the interval is nonempty and also that it is statistically interpretable as targeting $\Theta_I^*$. An obvious caveat is that, as with the related literature going back to \citet{White82}, the coverage target's substantive relevance may not be clear if the model is in fact misspecified. As \citet{AndrewsKwon19} elaborate, this has to be traded off against concerns with spurious precision.

While the coverage notion exactly mimics \citet{AndrewsKwon19}, the confidence interval will be quite different. It goes ``back to basics" in that, like early entries in the literature \citep{IM04,Stoye09}, it essentially just adds a certain number of standard errors to estimated bounds. An advantage is computational and conceptual simplicity; with test inversion intervals, critical values generally depend on $\theta$ even in this simple setting and therefore must be computed many times. However, the main motivation is that the new interval performs well. Its heuristic definition is as follows:
\begin{itemize}
\item Compute an interval
$$CI_{\Theta_I} \equiv \bigl[\hat{\theta}_L-\tfrac{\hat{\sigma}_L}{\sqrt{n}}\hat{c},\hat{\theta}_U+\tfrac{\hat{\sigma}_U}{\sqrt{n}}\hat{c}\bigr],$$
where $\hat{c}$ depends on $\alpha$ and $\hat{\rho}$; see Table \ref{tb:1}.
\item Also compute the estimator
\begin{eqnarray*}
\hat{\theta}^* &\equiv &\frac{\hat{\sigma}_U \hat{\theta}_L + \hat{\sigma}_L \hat{\theta}_U}{\hat{\sigma}_L+\hat{\sigma}_U}
\end{eqnarray*}
and confidence interval
\begin{eqnarray*}
CI_{\theta^*} &\equiv &\left[\hat{\theta}^*-\tfrac{\hat{\sigma}^*}{\sqrt{n}}\Phi^{-1}\left(1-\tfrac{\alpha}{2}\right) ,\hat{\theta}^*+\tfrac{\hat{\sigma}^*}{\sqrt{n}}\Phi^{-1}\left(1-\tfrac{\alpha}{2}\right) \right] \\
\hat{\sigma}^* &\equiv & \frac{\hat{\sigma}_L \hat{\sigma}_U \sqrt{2+2\hat{\rho}}}{\hat{\sigma}_L+\hat{\sigma}_U}.
\end{eqnarray*} 
\item Report the union $CI_{\Theta_I} \cup CI_{\theta^*}$.
\item We will not pre-estimate $\Delta$ but set it to its globally least favorable value. We will, however, anticipate the conservative bias ensuing from taking unions of intervals. This bias is easy to estimate; in particular, there is no parameter-on-the-boundary issue.
\item One might think that concentrating out $\Delta$ will be very conservative. It turns out that this is not so. In most cases, $\hat{c}=\Phi^{-1}(1-\alpha)$, i.e. we can just use the one-sided critical value, at least to extremely high simulation accuracy. If $\rho=0$ and for conventional coverage levels, this can be shown analytically.
\end{itemize}
The new confidence interval is obviously never empty; indeed, its length cannot drop below $2\hat{\sigma}^*\Phi^{-1}(1-\alpha/2)$. Its formal definition and theoretical justification are as follows.
\begin{definition}
The misspecification-adaptive confidence interval $CI_{MA}$ is
\begin{equation}
CI_{MA} \equiv \left[\hat{\theta}_L-\tfrac{\hat{\sigma}_L}{\sqrt{n}} \hat{c} ,  \hat{\theta}_U+\tfrac{\hat{\sigma}_U}{\sqrt{n}} \hat{c} \right] \cup \left[\hat{\theta}^*-\tfrac{\hat{\sigma}^*}{\sqrt{n}}\Phi^{-1}\left(1-\tfrac{\alpha}{2}\right),\hat{\theta}^*+\tfrac{\hat{\sigma}^*}{\sqrt{n}}\Phi^{-1}\left(1-\tfrac{\alpha}{2}\right)\right], \label{def:ci}
\end{equation}
where $\hat{c}$ is the unique value of $c$ solving
\begin{eqnarray}
& \inf_{\Delta \geq 0}&  \Pr \left(Z_1-\Delta-c \leq 0 \leq Z_2+c   \text{    or    } \left\vert Z_1+Z_2-\Delta \right\vert \leq  \sqrt{2+2\hat{\rho}}\Phi^{-1}\left(1-\tfrac{\alpha}{2}\right)  \right)  = 1-\alpha, \notag \\
&& \left(\begin{array}{c}Z_1\\Z_2\end{array}\right) \sim  N\left(\left(\begin{array}{c}0\\0\end{array}\right),\left(\begin{array}{cc}1 & \hat{\rho}\\ \hat{\rho} & 1 \end{array}\right)\right). \label{eq:ci}
\end{eqnarray}
If $\rho=0$ is known and $\sqrt{2}\Phi^{-1}(1-\alpha)\geq \Phi^{-1}(1-\alpha/2)$, just set $\hat{c}=\Phi^{-1}(1-\alpha)$.
\end{definition}
\begin{remark}
The condition that $\sqrt{2}\Phi^{-1}(1-\alpha)\geq \Phi^{-1}(1-\alpha/2)$ holds for $\alpha<.14$, i.e. for coverage levels of $86\%$ or higher.
\end{remark}
\begin{theorem} \label{thm:1}
The confidence interval $CI_{MA}$ achieves asymptotic coverage of $(1-\alpha)$.
\end{theorem}
\begin{proof}
See appendix \ref{sec:proof}.
\end{proof}
\begin{table}
\centering
\begin{tabular}
{llllllll}
$\bm{\rho}$    & $\leq$0.8 & 0.85 &   0.9 & 0.95 & 0.98 & 0.99 & 1.0
\\[0.75ex]
\hline\\[-1.5ex]
$\bm{\alpha=.1}$ &  1.28 & 1.29 &  1.31 & 1.36 & 1.44 & 1.54&  1.64\\
$\bm{\alpha=.05}$  &  1.64 &1.65 &  1.65 & 1.70 & 1.76 & 1.81 & 1.96\\
$\bm{\alpha=.01}$  &  2.33 & 2.33 & 2.33 & 2.34 & 2.40 & 2.43 & 2.58
\end{tabular}
\caption{Critical values obtained by concentrating out $\Delta \in [0,\infty)$ for different coverages and correlations. For $\rho \leq 0.8$, further simulations corroborate the one-sided critical value as exact solution.}
\label{tb:1}
\end{table}
Expression \eqref{eq:ci} is numerically evaluated for different values of $\rho$ and target coverages in Table \ref{tb:1}. In particular, simulation with very high accuracy suggests that $\hat{c}$ is just the one-sided critical value for $\rho$ up to at least $.8$; it then gradually increases toward the two-sided critical value, which is easily seen to solve \eqref{eq:ci} for $\hat{\rho}=1$.\footnote{The table was generated by gridding and using $B=4000000$ simulations. This is feasible on a run-of-the-mill netbook. The relevant simulation error is the coverage error at the suggested $\hat{c}$. Given $B$, it will be much smaller than what is routinely accepted in simulation-based, e.g. bootstrap, inference. For $\rho \leq .8$, further simulations establish to high accuracy that coverage is first increasing and then decreasing in $\Delta$ and minimized as $\Delta \to \infty$, the same feature that is analytically proved for $\rho=0$ and which justifies $\hat{c}=\Phi^{-1}(1-\alpha)$.}  

%The relation between \eqref{eq:ci} and the verbal definition of $\hat{c}$ is not obvious. One might expect $\hat{c}$ to depend on $(\hat{\sigma}_L,\hat{\sigma}_U)$ and also to explicitly control coverage uniformly over $\theta \in [\theta_L,\theta_U]$. Expression \eqref{eq:ci} is much simpler than that. The reason is that $(\hat{\sigma}_L,\hat{\sigma}_U)$ cancel out of expressions and that it is without loss of generality to set $\theta=\theta_U$. To be clear, these simplifications are not obvious (they correspond to step 2 of the proof) and owe to careful construction of $CI_{MA}$.

\begin{remark} \label{rem:1}
Except for large positive $\rho$, infimal coverage of $(1-\alpha)$ is attained in the limit as $\Delta \to \infty$. For finite $\Delta$, $CI_{MA}$ is therefore nominally conservative.

In principle, one could try to capture this by concentrating out $\Delta$ over a more limited range, e.g. over a $(1-.1\alpha)$-confidence interval with Bonferroni adjustment of subsequent inference. This could in principle lead to $\hat{c}<\Phi^{-1}(1-\alpha)$. I do not advocate it because numerically, the infimum in \eqref{eq:ci} is well approximated for surprisingly small values of $\Delta$. Therefore, the ``inferential cost" of a pre-test, whether through adjustment of second-stage test size or through reliance on a tuning parameter, would typically not be recovered.     
\end{remark}

\begin{remark} \label{rem:uniform}
The literature on partial identification often focuses on uniform inference. This is because na\"ive inference methods may fail in cases of interest, e.g. as one approaches point identification. To prevent this, the literature has an informal requirement that inference be uniform over delicate nuisance parameters like (in this paper) $\Delta$; see \citet[][Section 4.3.2]{MolinariHOE} for further discussion. $CI_{MA}$ is obviously uniform in this sense because $\Delta$ (and also the position of $\theta$ in $\Theta_I$) is set to its globally least favorable value.

To formally claim that inference is uniform over a large class of data generating processes, one would furthermore have to strengthen Assumption \ref{as:1} so that consistency and asymptotic normality of bound estimators hold in a uniform sense. The exact nature of such strengthenings, and low-level assumptions that achieve them, are well understood \citep{AS10,RS08} and are omitted for brevity.
\end{remark}
\begin{remark} \label{rem:im}
The notable difference in setting to \citet{IM04} is the absence of an implicit superefficiency condition on $\hat{\Delta}$ near true value $0$. That condition turns out to obtain if (and, in practice, only if) $\hat{\theta}_U \geq \hat{\theta}_L$ \textit{by construction} \citep[][Lemma 3]{Stoye09}. This case is empirically relevant: It applies to most missing-data bounds and also bounds that rely on different truncations of observed probability measures \citep{HM95,Lee09}, unless further refinements turn these into intersection bounds. If it obtains and other regularity conditions hold, the confidence interval in \citet{IM04} is valid, is expected to be rather efficient for small $\Delta$ (because it uses superconsistency of $\hat{\Delta}$), and will obviously never be empty. Not coincidentally, this case is also characterized by the possibility of $\rho \approx 1$; indeed, that is how superconsistency of $\hat{\Delta}$ arises. Whether this case applies can be ascertained before seeing any data, and I strongly suggest that users do so.
\end{remark}
\begin{remark}
I follow the bulk of the literature in focusing on uniform coverage of $\theta \in \Theta_I^*$. The procedure is easily adapted to coverage of the entire set $\Theta_I^*$. Note that, by a Bonferroni argument, a critical value of $\hat{c}=\Phi^{-1}(1-\alpha/2)$ would always do, and also that (as can be seen from considering large $\Delta$) only a large negative value of $\hat{\rho}$ would cause $\hat{c}$ to be appreciably lower. 
\end{remark}
The proof of Theorem \ref{thm:1} contains three steps. First, it is relatively routine to show that $CI_{MA}$ would be valid if, in line with the heuristic definition, expression \eqref{eq:ci} explicitly took the infimum also over values of $(\sigma_L,\sigma_U)$ as well as $\theta \in \Theta_I$. In a second step, we can concentrate out all of these. In particular, one can restrict attention to one of $\theta=\theta_L$ or $\theta=\theta_U$; expression \eqref{eq:ci} arbitrarily chooses the latter. This finding is not obvious: For given $\Delta$, coverage is \textit{not} equally minimized at the interval's endpoints; it is only that the corresponding infima over $\Delta \in [0,\infty)$ are the same. As final flourish in this step, it turns out that asymptotic coverage at $\theta_U$ depends on $(\Delta,\sigma_L,\sigma_U)$ only through $\Delta/\sigma_L$. For the purpose of evaluating worst-case coverage over $\Delta \geq 0$, we can therefore set both standard deviations to $1$.

The final, and by far most delicate, step is that if $\rho=0$, coverage is provably minimized as $\Delta \to \infty$, justifying use of the one-sided critical value $\hat{c}=\Phi^{-1}(1-\alpha)$. To appreciate this claim, consider again the two components of $CI_{MA}$ in \eqref{def:ci}. For $\alpha=.05$, the left-hand interval's coverage for either $\theta_L$ or $\theta_U$ may be as low as $.9$ if $\Delta=0$ and approach $.95$ \textit{from below} as $\Delta \to \infty$. The right-hand interval's coverage of these values is $.95$ at $\Delta=0$ (where both coincide with $\theta^*$) but rapidly decreases to $0$ as $\Delta$ increases. That these effects aggregate to coverage uniformly above $.95$ is far from obvious and heavily relies on specific features of the bivariate Normal distribution.     

Numerically, the final step extends to moderate $\rho$ (see again Table \ref{tb:1}), and the proof uses conservative bounds. Some analytic result of higher generality might, therefore, be available. However, for large positive $\rho$, coverage is minimized at small positive $\Delta$. Therefore, if $\rho$ is unknown, estimating it cannot be avoided. In particular, in view of Table \ref{tb:1}, a pre-test for ``small enough" $\rho$ would be counterproductive: Since $\hat{c}$ as a function of $\rho$ is mostly completely flat, one would be unlikely to recover the inferential cost (in the sense of Remark \ref{rem:1}) of the pre-test.

\begin{figure}[t]
\begin{subfigure}[b]{.5\linewidth}
			\begin{adjustbox}{max width=\textwidth}
\includegraphics[trim=4.5cm 8cm 4.5cm 8.8cm,clip=true,scale=.5]{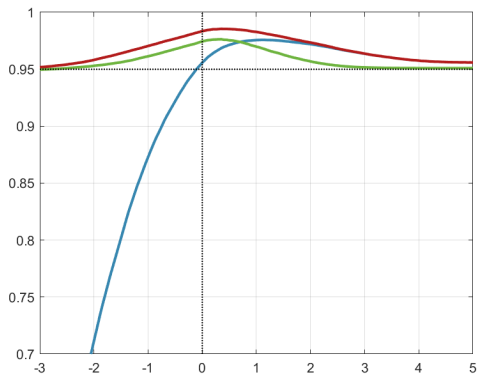}
\end{adjustbox}
\caption{Coverage when $\rho=0$.}
		\end{subfigure}%
		\begin{subfigure}[b]{.5\linewidth}
		\begin{adjustbox}{max width=\textwidth}
\includegraphics[trim=4.5cm 8cm 4.5cm 8.8cm,clip=true,scale=.5]{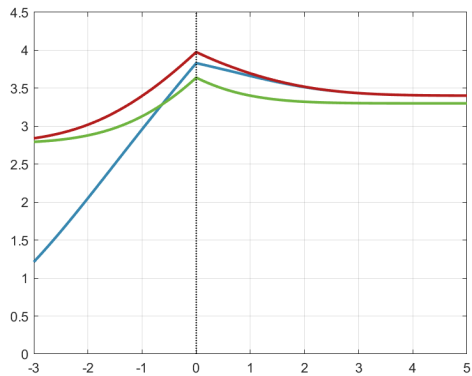}
\end{adjustbox}
\caption{Expected length when $\rho=0$.}
\end{subfigure}
\begin{subfigure}[b]{.5\linewidth}
			\begin{adjustbox}{max width=\textwidth}
\includegraphics[trim=4.5cm 8cm 4.5cm 8.8cm,clip=true,scale=.5]{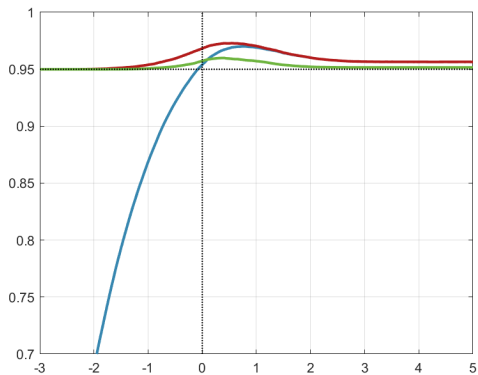}
\end{adjustbox}
\caption{Coverage when $\rho=0.7$.}
\end{subfigure}
		\begin{subfigure}[b]{.5\linewidth}
		\begin{adjustbox}{max width=\textwidth}
\includegraphics[trim=4.5cm 8cm 4.5cm 8.8cm,clip=true,scale=.5]{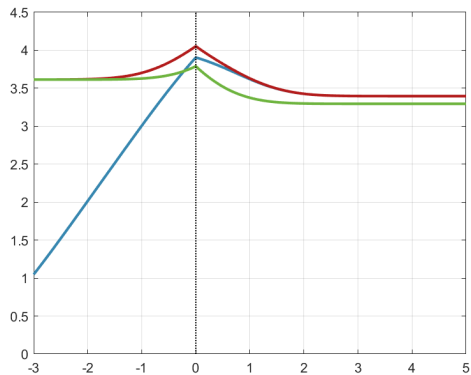}
\end{adjustbox}
\caption{Expected length when $\rho=0.7$.}
\end{subfigure}
\caption{Coverage (left panels) and expected length (right panels; length of true interval is subtracted) of $CI_{TI}$ (blue), $CI_{TI} \cup CI_{\theta^*}$ (red) and the new proposal $CI_{MA}$ (green). Horizontal axis is $\Delta=\theta_U-\theta_L$; negative values indicate increasing misspecification.  Nominal coverage is $95\%$ and is indicated by a black horizontal line.}
\label{fig:1}
\end{figure}

\begin{figure}[t]
\begin{subfigure}[b]{.5\linewidth}
			\begin{adjustbox}{max width=\textwidth}
\includegraphics[trim=4.5cm 8cm 4.5cm 8.8cm,clip=true,scale=.5]{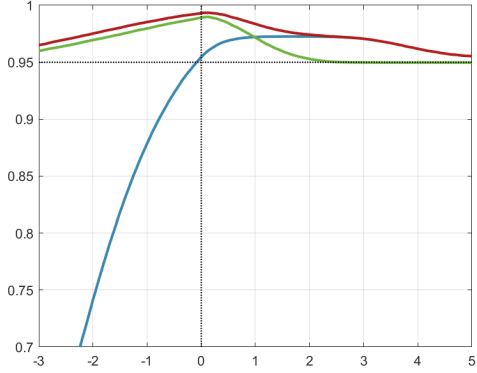}
\end{adjustbox}
\caption{Coverage when $\rho=-0.7$.}
\end{subfigure}
		\begin{subfigure}[b]{.5\linewidth}
		\begin{adjustbox}{max width=\textwidth}
\includegraphics[trim=4.5cm 8cm 4.5cm 8.8cm,clip=true,scale=.5]{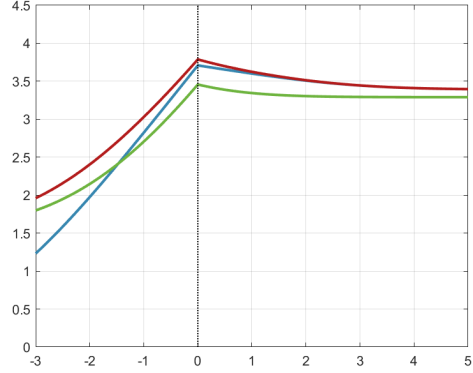}
\end{adjustbox}
\caption{Expected length when $\rho=-0.7$.}
\end{subfigure}
\begin{subfigure}[b]{.5\linewidth}
			\begin{adjustbox}{max width=\textwidth}
\includegraphics[trim=4.5cm 8cm 4.5cm 8.8cm,clip=true,scale=.5]{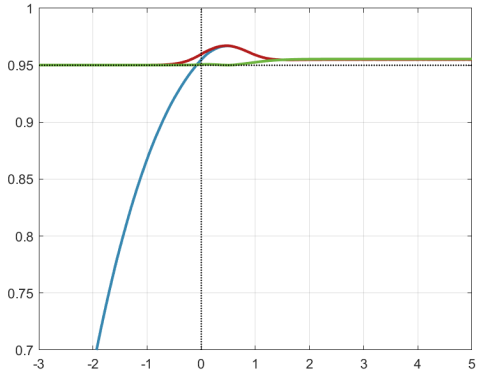}
\end{adjustbox}
\caption{Coverage when $\rho=0.95$.}
\end{subfigure}
		\begin{subfigure}[b]{.5\linewidth}
		\begin{adjustbox}{max width=\textwidth}
\includegraphics[trim=4.5cm 8cm 4.5cm 8.8cm,clip=true,scale=.5]{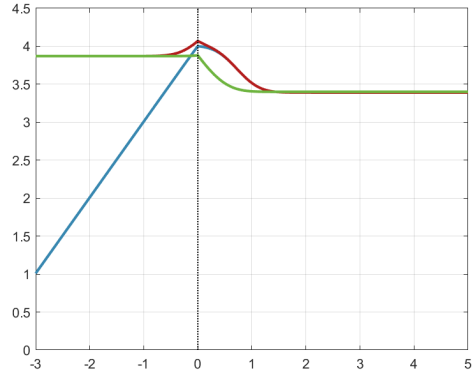}
\end{adjustbox}
\caption{Expected length when $\rho=0.95$.}
\end{subfigure}
\caption{Continuation of Figure \ref{fig:1}. The last case ($\rho=.95$) illustrates a setting where $\Delta \to \infty$ is not least favorable and where $\hat{c}>1.64$.}
\label{fig:2}
\end{figure}

\section{Numerical Illustration} \label{sec:numerical}
Figures \ref{fig:1} and \ref{fig:2} compare $CI_{MA}$ with a test inversion interval $CI_{TI}$ that arguably reflects the state of the established literature.\footnote{The interval closely follows \citet{RSW14}; other established methods \citep{AS10,AB12,Bugni10,Canay10} would inform similar constructions. As of writing of this manuscript, at least two  rather distinct (from the preceding and from each other) proposals are in the pipeline \citep{APR19,CS19}. Both invert a test and can be empty; \citet{APR19} also has a tuning parameter. They are compared in \citet{CS19}. A comparison of all these approaches in simple examples might be worthwhile.} It inverts a test of $H_0:\theta \leq \theta_U,\theta \geq \theta_L$ by taking the maximum (studentized) violation as test statistic, i.e. the same test statistic that generally implies $\Theta_I^*$ as pseudotrue identified set. The critical value is based on a pre-test --specifically, a one-sided $(.1\alpha)$-Wald test-- that potentially discards one of the inequality constraints as nonbinding. Depending on the pre-test's result, the critical value is then either a simple one-sided critical value or computed by a simulation that takes $\rho$ into account. In either case, the second-stage test is of size $.9\alpha$, so that the pre-test is accounted for by Bonferroni correction. The resulting test is inverted, and the critical value is recomputed, as $\theta$ changes, making the interval considerably shorter than early entries in the literature \citep{IM04,Stoye09}. Compared to $CI_{MA}$, test inversion adds orders of magnitude of computational cost, though at a very low absolute level. I abstract from asymptotic approximation by drawing estimators straight from limiting distributions and taking $(\sigma_L,\sigma_U,\rho)$ to be known. Interval length $\Delta$ is denominated in estimator standard errors because $\sqrt{n}\sigma_L=\sqrt{n}\sigma_U=1$ throughout.

The comparison is extended into the misspecified range by letting $\Delta$ take on negative values. The test inversion interval obviously undercovers in that range. To clarify comparisons, I also compute $CI_{TI} \cup CI_{\theta^*}$. Recall that $CI_{MA}$ can be loosely intuited as refining this construction by adjusting the critical value to account for union-taking. Nominal coverage is $95\%$ throughout.

Figure \ref{fig:1} illustrates the results for $\rho=0$ (top panels) and $\rho=.7$ (bottom panels); Figure \ref{fig:2} extends the exercise to $\rho=-.7$ and finally to $\rho=.95$. The last case is arguably contrived but serves to illustrate that $\Delta\to \infty$ is not always least favorable. By the same token, this is the only case in which $\hat{c}>\Phi^{-1}(.95)$.

With one caveat discussed below, the figures suggest dominating performance of $CI_{MA}$: It is shorter, and this is also reflected in more precise size control and thereby more power of the implied test. The advantage is especially apparent for small positive $\Delta$. What happens here is that the correction provided by $CI_{\theta^*}$ allows $CI_{MA}$ to transition to just adding $1.64$ standard errors considerably more quickly than a pre-test could justify. Indeed, for $\rho\leq .4$, this transition occurs at a \textit{negative} estimated interval length $\hat{\Delta}$; that is, $CI_{MA}$ just adds $1.64$ standard errors to bounds estimates whenever these are ordered in the expected way. The slight advantage of $CI_{MA}$ for large $\Delta$ reflects that $CI_{TI}$ accounts for a pre-test.

One might wonder how \citet{AndrewsKwon19} would perform in the example. While the exact answer depends on choice of multiple tuning parameters, some qualitative considerations are as follows. Their interval starts from $CI_{TI}$ and expands it in order to avoid spurious precision.\footnote{\citet{AndrewsKwon19} implement $CI_{TI}$ through \citet{AS10} but point out that \citet{RSW14} could be used instead. The difference will be small in the present setting.} As a result, it will be bounded from below in both length and coverage by the blue curves in Figures \ref{fig:1} and \ref{fig:2}. In an initial refinement, \citet{AndrewsKwon19} form the union between $CI_{TI}$ and a never-empty confidence interval. Their preferred confidence interval does this only if an additional pre-test fails to reject misspecification. While this mitigates the effect of expanding $CI_{TI}$, the final confidence interval still contains $CI_{TI}$ and considerably exceeds it for small positive $\Delta$ (see their Section 8.1, whose setting resembles the present one). This will obviously be reflected in its statistical performance. Conversely, an intriguing feature of $CI_{MA}$ is that it ``spends" the ``coverage capital" gained from ensuring nonemptiness by being shorter than $CI_{TI}$ for interesting values of $\Delta$. In fairness to \citet{AndrewsKwon19}, it appears far from obvious how to implement such a feature in their much more general setting.

The advantage of $CI_{MA}$ fades out, and even reverses, in the special case where $\rho\to 1$ but \textit{not} $\Delta \to 0$. In that limit, $\hat{c}$ will converge to the two-sided critical value, whereas a pre-test will eventually recommend a one-sided test. While such scenarios can obviously be simulated, they arguably are contrived. The possibility of high $\rho$ and correspondingly precise estimation of $\Delta$ is empirically relevant, but it corresponds to the superefficiency case discussed in Remark \ref{rem:im} and therefore to small $\Delta$ as well as to a case distinction that can be decided in pre-data analysis. Also, one could in principle fix this issue by layering a pre-test on top of $CI_{MA}$; however, as general advice in this matter, I stand by Remark \ref{rem:1}.

\section{Empirical Application} \label{sec:empirical}
\Citet{Haushofer} estimate upper and lower bounds on behavioral parameters from different treatments in a between-subjects design, meaning that estimators are uncorrelated. At the same time, bounds can and did in fact invert, triggering an inquiry by the authors that led to the present paper.

\begin{table}
\centering
\begin{tabular}{lcccc}
\textbf{Game} &  $\bm{[\hat{\theta}_L,\hat{\theta}_U]}$ & $\bm{CI_{MA}}$ & $\bm{CI_{TI}}$ & \textbf{rel. length} 
\\[0.75ex]
\hline\\[-1.5ex]
Ambiguity Aversion &  [0.499,0.557] & [0.459,0.597] & [0.458,0.598] & 0.97 \\
Effort: 1 cent bonus &  [0.469,0.484] & [0.448,0.503] & [0.448,0.504] & 0.97\\
Effort: 0 cent bonus$^*$ &  [0.343,0.331] & [0.318,0.356] & [0.315,0.358] & 0.91\\
Lying$^{**}$ &  [0.530,0.537] & [0.512,0.556] & [0.508,0.560] & 0.83\\
Time$^{**}$ &  [0.766,0.770] & [0.722,0.814] & [0.712,0.824] & 0.82\\
Trust Game 1  & [0.430,0.455] & [0.388,0.493] & [0.387,0.495] & 0.96\\
Trust Game 2 & [0.348,0.398] & [0.328,0.426] & [0.327,0.427] & 0.97\\
Ultimatum Game 1& [0.443,0.470] & [0.422,0.493] & [0.422,0.494] & 0.97\\
Ultimatum Game 2 & [0.362,0.413] & [0.342,0.436] & [0.341,0.436] & 0.97
\end{tabular}
\caption{Confidence intervals applied to data in \citet[][compare select columns of their Table 1]{Haushofer}. Relative length refers to relative (of $CI_{MA}$ over $CI_{TI}$) excess length beyond $\max\{\hat{\Delta},0\}$. Of special interest: Case (*) has inverted bound estimators, displayed with abuse of interval notation. Cases (**) are short (near point identified) estimated intervals.}
\label{tb:2}
\end{table}

Table \ref{tb:2} displays estimated bounds, $CI_{MA}$, and $CI_{TI}$ for selected instances of the ``weak bounds" data. This refers to a baseline setting before inducing experimenter demand. For more details, I refer to \citet{Haushofer}, particularly Figure 1 and corresponding explanations. The last column divides the length of $CI_{MA}$ by the length of $CI_{TI}$, subtracting $\max\{\hat{\Delta},0\}$ from both. Both intervals make full use of $\rho=0$ being known.

The comparison is between $CI_{MA}$ and $CI_{TI}$; obviously, $CI_{TI}\cup CI_{\theta^*}$ would be larger than $CI_{TI}$. The data include one case (*) where bound estimators are inverted and where ex post, $CI_{MA}=CI_{\theta^*}$.\footnote{This case would not have led any specification test to reject the model, even before taking multiple hypothesis testing into account.} There are also two cases (**) of short estimated intervals (relative to standard errors), i.e. of near point identification. Because $CI_{MA}$ cannot be empty, one might have conjectured it to be the longer one in these cases. In fact, it is noticeably shorter in all of them -- the effect of ``spending coverage capital" from the nonemptiness correction dominates. In all other cases, both intervals effectively add $1.64$ standard errors.\footnote{In those cases, the small differences favoring $CI_{MA}$ reflect Bonferroni adjustment for pre-tests, i.e. the specifics of \citet{RSW14}. In cases where $[\theta_L,\theta_U]$ is obviously ``long," researchers will in practice be tempted to claim an asymptotic pre-test and just use $1.64$ standard errors.}

\section{Conclusion}\label{sec:conclusion}
For a simple, but empirically relevant, partial identification problem, I propose a confidence interval that has competitive size control and length including in the misspecified case, while being extremely easy to compute. The most striking finding is that in many cases, a seemingly crude fix to a nominal $90\%$ confidence interval ensures $95\%$ coverage at little cost in terms of interval length and with practically zero computation. Simulations are encouraging, and the confidence interval improves on current best practice in application to recent lab experiments.

The approach is complementary to \citet{AndrewsKwon19}, from whom I take the broad motivation as well as the novel coverage requirement. Of course, their approach applies far beyond the present paper's simple setting. On the other hand, it has several tuning parameters and expands a conventional confidence interval, whereas the present proposal is tuning parameter free and compensates for expanding the conventional interval by reducing its standalone nominal coverage. A question of obvious interest, but also beyond my current reach, is whether this last feature can be usefully generalized. As it stands, the present proposal is limited to a specific setting but appears both practical and powerful when that setting obtains.

\newpage
\appendix
\section{Proof of Theorem \ref{thm:1}} \label{sec:proof}
Validity in case of $\Delta<0$ is obvious -- only coverage of $\theta^*$ is required in this case, and $CI_{\theta^*}$ achieves that by itself. Since for $\Delta \geq 0$, we have $\Theta_I^*=\Theta_I$, it remains to show coverage of $\theta \in [\theta_L,\theta_U]$ assuming that $\theta_U \geq \theta_L$. For the remainder of this proof, express the true value of $\theta$ as $\theta=\lambda \theta_U + (1-\lambda) \theta_L$ for some $\lambda \in [0,1]$. Consider initially the idealized confidence interval $CI^*_{MA}$, which is just like $CI_{MA}$ except that, rather than by \eqref{eq:ci}, a critical value $c^*$ is defined by setting $\inf_{\Delta\geq 0,\lambda \in [0,1]}\Pr (E_{\Delta,\lambda,c^*}) = 1-\alpha$, where
\begin{eqnarray}
E_{\Delta,\lambda,c} &=& \left\{ Z^*_1 - \frac{\lambda}{\sigma_L} \Delta \leq c \cap Z^*_2 + \frac{1-\lambda}{\sigma_U} \Delta \geq - c \right\} \label{eq:calibrate_proof} \\
& \cup &\left\{ Z^*_1+  Z^*_2 + \left( \frac{1-\lambda}{\sigma_U}- \frac{\lambda}{\sigma_L} \right) \Delta  \in \left[-\sqrt{2+2\rho} \Phi^{-1}\left(1-\tfrac{\alpha}{2}\right),\sqrt{2+2\rho} \Phi^{-1}\left(1-\tfrac{\alpha}{2}\right)\right] \right\}, \notag  \\
&& \left(\begin{array}{c}Z^*_1\\Z^*_2\end{array}\right) \sim  N\left(\left(\begin{array}{c}0\\0\end{array}\right),\left(\begin{array}{cc}1 & \rho\\ \rho & 1 \end{array}\right)\right).\notag
\end{eqnarray}
Note two differences to \eqref{eq:ci}: The construction explicitly minimizes over both $\Delta$ and $\lambda$, and it is infeasible in that population values of $(\sigma_L,\sigma_U,\rho)$ are used.

Step 1 of the proof establishes validity of $CI^*_{MA}$. Step 2 shows that $\lambda$ can always be set to $1$, transforming the above into \eqref{eq:ci}. Step 3 establishes that if $\rho=0$, one can furthermore take the limit as $\Delta \to \infty$. The argument that $(\sigma_L,\sigma_U,\rho)$ can be replaced with consistent estimators is omitted for brevity. 

\paragraph{Step 1: Validity of $CI_{MA}^*$.}
Write   
\begin{equation*}
CI^*_{MA} = \left[\hat{\theta}_L-\frac{\sigma_L}{\sqrt{n}} c^*,\hat{\theta}_U+\frac{\sigma_U}{\sqrt{n}} c^*\right] \cup \left[\frac{\sigma_L \hat{\theta}_U + \sigma_U \hat{\theta}_L}{\sigma_L + \sigma_U} -  \frac{\sigma^*}{\sqrt{n}}\Phi^{-1}\left(1-\tfrac{\alpha}{2}\right), \frac{\sigma_L \hat{\theta}_U + \sigma_U \hat{\theta}_L}{\sigma_L + \sigma_U} + \frac{\sigma^*}{\sqrt{n}}\Phi^{-1}\left(1-\tfrac{\alpha}{2}\right) \right],
\end{equation*}
where $\sigma^*\equiv \sqrt{2+2\rho}\sigma_L \sigma_U / (\sigma_L+\sigma_U)$ is the asymptotic standard deviation of $\sqrt{n}(\lambda^*\hat{\theta}_U+(1-\lambda^*)\hat{\theta}_L-\theta^*)$ and $\lambda^*\equiv \sigma_L/(\sigma_L+\sigma_U)$ is the mixture weight characterizing $\theta^*$.

Define also standardized estimation errors
\begin{equation*}
(\bar{\varepsilon}_L,\bar{\varepsilon}_U)\equiv \sqrt{n}\left(\frac{\hat{\theta}_L-\theta_L}{\sigma_L},\frac{\hat{\theta}_U-\theta_U}{\sigma_U}\right).
\end{equation*}
We have that $\theta \in CI^*_{MA}$ if either
\begin{eqnarray*}
&& \hat{\theta}_L-\frac{\sigma_L}{\sqrt{n}} c^* \leq \lambda \theta_U + (1-\lambda) \theta_L \leq \hat{\theta}_U+\frac{\sigma_U}{\sqrt{n}} c^* \\
& \Longleftrightarrow & \hat{\theta}_L-\theta_L \leq \lambda \Delta + \frac{\sigma_L}{\sqrt{n}} c^*, ~~~~ \hat{\theta}_U-\theta_U \geq -(1-\lambda)\Delta - \frac{\sigma_U}{\sqrt{n}} c^* \\
& \Longleftrightarrow & \bar{\varepsilon}_L \leq \frac{\lambda}{\sigma_L}\sqrt{n} \Delta + c^* , ~~~~ \bar{\varepsilon}_U \geq -\frac{1-\lambda}{\sigma_U}\sqrt{n}\Delta - c^*
\end{eqnarray*}
or
\begin{eqnarray*}
&& \frac{\sigma_L \hat{\theta}_U + \sigma_U \hat{\theta}_L}{\sigma_L+\sigma_U} -(\lambda \theta_U + (1-\lambda) \theta_L) \in \left[- \frac{\sigma^*}{\sqrt{n}} \Phi^{-1}\left(1-\tfrac{\alpha}{2}\right), \frac{\sigma^*}{\sqrt{n}} \Phi^{-1}\left(1-\tfrac{\alpha}{2}\right)\right] \\
&\Longleftrightarrow & \frac{\sigma_L \left(\theta_U+\frac{\sigma_U \bar{\varepsilon}_U}{\sqrt{n}}\right) + \sigma_U \left(\theta_L+\frac{\sigma_L \bar{\varepsilon}_L}{\sqrt{n}}\right)}{\sigma_L+\sigma_U} -(\lambda \theta_U + (1-\lambda) \theta_L) \in \left[- \frac{\sigma^*}{\sqrt{n}} \Phi^{-1}\left(1-\tfrac{\alpha}{2}\right), \frac{\sigma^*}{\sqrt{n}} \Phi^{-1}\left(1-\tfrac{\alpha}{2}\right)\right] \\
&\Longleftrightarrow & \frac{\sigma_L \sigma_U}{\sigma_L + \sigma_U} (\bar{\varepsilon}_L+ \bar{\varepsilon}_U) + \sqrt{n} \left(\frac{\sigma_L \theta_U + \sigma_U \theta_L}{\sigma_L+\sigma_U}   -(\lambda \theta_U + (1-\lambda) \theta_L)\right)  \in \Bigl[-\sigma^* \Phi^{-1}\left(1-\tfrac{\alpha}{2}\right),\sigma^* \Phi^{-1}\left(1-\tfrac{\alpha}{2}\right)\Bigr] \\
&\Longleftrightarrow & \bar{\varepsilon}_L+ \bar{\varepsilon}_U +\sqrt{n}\frac{\sigma_L \theta_U + \sigma_U \theta_L - (\sigma_L+\sigma_U)(\lambda \theta_U + (1-\lambda) \theta_L)}{\sigma_L \sigma_U} \\
&& ~~~~~~~~~~~~~~~~~~~~~~~~~~~~  \in \left[-\sqrt{2+2\rho} \Phi^{-1}\left(1-\tfrac{\alpha}{2}\right),\sqrt{2+2\rho} \Phi^{-1}\left(1-\tfrac{\alpha}{2}\right)\right] \\
 &\Longleftrightarrow & \bar{\varepsilon}_L + \bar{\varepsilon}_U + \left( \frac{1-\lambda}{\sigma_U}-\frac{\lambda}{\sigma_L} \right) \sqrt{n}\Delta  \in \left[-\sqrt{2+2\rho} \Phi^{-1}\left(1-\tfrac{\alpha}{2}\right),\sqrt{2+2\rho} \Phi^{-1}\left(1-\tfrac{\alpha}{2}\right)\right].
\end{eqnarray*}
In sum,
\begin{eqnarray*}
&& \Pr(\theta \in CI^*_{MA}) \\
&=& \Pr \left( \left\{ \bar{\varepsilon}_L - \frac{\lambda}{\sigma_L} \sqrt{n}\Delta \leq c^* \cap \bar{\varepsilon}_U + \frac{1-\lambda}{\sigma_U} \sqrt{n}\Delta \geq - c^* \right\}\right.  \\
&& \cup \left. \left\{ \bar{\varepsilon}_L + \bar{\varepsilon}_U + \left( \frac{1-\lambda}{\sigma_U}-\frac{\lambda}{\sigma_L} \right) \sqrt{n}\Delta   \in \left[-\sqrt{2+2\rho} \Phi^{-1}\left(1-\tfrac{\alpha}{2}\right),\sqrt{2+2\rho} \Phi^{-1}\left(1-\tfrac{\alpha}{2}\right)\right] \right\}\right) \\
&\to & \Pr \left( \left\{ Z_1^* - \frac{\lambda}{\sigma_L} \sqrt{n}\Delta \leq c^* \cap Z_2^* + \frac{1-\lambda}{\sigma_U} \sqrt{n}\Delta \geq - c^* \right\}\right.  \\
&& \cup \left. \left\{ Z_1^*  +  Z_2^* + \left( \frac{1-\lambda}{\sigma_U}- \frac{\lambda}{\sigma_L}\right)  \sqrt{n}\Delta   \in \left[-\sqrt{2+2\rho} \Phi^{-1}\left(1-\tfrac{\alpha}{2}\right),\sqrt{2+2\rho} \Phi^{-1}\left(1-\tfrac{\alpha}{2}\right)\right] \right\}\right) \\
&\geq &  \inf_{\Delta \geq 0,\lambda \in [0,1]}  \Pr \left( \left\{ Z_1^* - \frac{\lambda}{\sigma_L} \Delta \leq c^* \cap Z_2^* + \frac{1-\lambda}{\sigma_U} \Delta \geq - c^* \right\}\right.  \\
&& \cup \left. \left\{ Z_1^*  +  Z_2^* + \left( \frac{1-\lambda}{\sigma_U}- \frac{\lambda}{\sigma_L}\right)  \Delta   \in \left[-\sqrt{2+2\rho} \Phi^{-1}\left(1-\tfrac{\alpha}{2}\right),\sqrt{2+2\rho} \Phi^{-1}\left(1-\tfrac{\alpha}{2}\right)\right] \right\}\right) \\
&= & 1-\alpha,
\end{eqnarray*}
where the convergence uses Assumption $\ref{as:1}$ and the next step uses the definition of $c^*$ and also observes that, since we take an infimum over $\Delta \geq 0$, we can drop the $\sqrt{n}$ premultiplying $\Delta$.

\paragraph{Step 2: Concentrating out $\lambda$.}
We first concentrate out $\lambda$, for which $\{0,1\}$ are equally least favorable if $\Delta$ is unrestricted. To see this, consider the reparameterization
\begin{equation}
(X_1,X_2) \equiv \left(\frac{Z_2^*+Z_1^*}{\sqrt{2}},\frac{Z_2^*-Z_1^*}{\sqrt{2}}\right) \Longleftrightarrow   (Z_1^*,Z_2^*) = \left(\frac{X_1-X_2}{\sqrt{2}},\frac{X_1+X_2}{\sqrt{2}}\right) \label{eq:reparameterize}
\end{equation}
and observe that $(X_1,X_2)$ are uncorrelated. Simple algebra yields
\begin{eqnarray*}
E_{\Delta,\lambda,c} &=& \left\{ X_1-X_2 - \frac{\lambda}{\sigma_L}\sqrt{2} \Delta \leq \sqrt{2} c \cap X_1+X_2 + \frac{1-\lambda}{\sigma_U} \sqrt{2} \Delta \geq - \sqrt{2} c \right\} \notag \\
& \cup & \left\{ X_1 + \left(\frac{1-\lambda}{\sigma_U}  - \frac{\lambda}{\sigma_L}\right) \frac{\Delta}{\sqrt{2}}   \in \left[- \sqrt{1+\rho}\Phi^{-1}\left(1-\tfrac{\alpha}{2}\right), \sqrt{1+\rho}\Phi^{-1}\left(1-\tfrac{\alpha}{2}\right)\right] \right\}.
\end{eqnarray*}
Consider minimizing $\Pr(E_{\Delta,\lambda,c} )$ subject to the constraint that
\begin{equation*}
\Delta = \frac{\sigma_L \sigma_U}{\lambda \sigma_U+(1-\lambda)\sigma_L }  \beta  
\end{equation*}
for some fixed value $\beta \geq 0$. This is without loss of generality since one can minimize over $\beta$ in a second step and every value of $(\Delta,\lambda)\in [0,\infty) \times [0,1]$ is consistent with some $\beta \geq 0$. Also rearranging expressions to be of form ``$\ldots\leq X_1 \leq \dots$", one can write
\begin{eqnarray*}
&& E_{\Delta,\lambda,c}\Vert_{\Delta = \frac{\sigma_L \sigma_U}{\lambda \sigma_U+(1-\lambda)\sigma_L}  \beta  } \\
&=& \left\{ - X_2  -\sqrt{2}c - \frac{(1-\lambda)\sigma_L}{\lambda \sigma_U+(1-\lambda)\sigma_L } \sqrt{2} \beta \leq X_1 \leq X_2 + \sqrt{2}c  + \frac{\lambda \sigma_U}{\lambda \sigma_U+(1-\lambda)\sigma_L } \sqrt{2}\beta  \right\} \notag \\
& \cup & \left\{  \frac{\lambda \sigma_U-(1-\lambda)\sigma_L}{\lambda \sigma_U+(1-\lambda)\sigma_L} \times \frac{\beta}{\sqrt{2}}- \sqrt{1+\rho}\Phi^{-1}\left(1-\tfrac{\alpha}{2}\right) \leq X_1 \leq  \frac{\lambda \sigma_U-(1-\lambda)\sigma_L}{\lambda \sigma_U+(1-\lambda)\sigma_L} \times \frac{\beta}{\sqrt{2}} + \sqrt{1+\rho}\Phi^{-1}\left(1-\tfrac{\alpha}{2}\right) \right\}
\end{eqnarray*}
and therefore
\begin{eqnarray*}
&& \Pr(E_{\Delta,\lambda,c}|X_2=x_2)\Vert_{\Delta = \frac{\sigma_L \sigma_U}{\lambda \sigma_U+(1-\lambda)\sigma_L}  \beta  } \\
&=& \Pr\left(X_1 \in \left[-x_2  -\sqrt{2}c - \frac{(1-\lambda)\sigma_L}{\lambda \sigma_U+(1-\lambda)\sigma_L } \sqrt{2}\beta, x_2 + \sqrt{2}c  + \frac{\lambda \sigma_U}{\lambda \sigma_U+(1-\lambda)\sigma_L } \sqrt{2}\beta  \right] \right. \\ 
&\cup &\left. \left[  \frac{\lambda \sigma_U-(1-\lambda)\sigma_L}{\lambda \sigma_U+(1-\lambda)\sigma_L} \times \frac{\beta}{\sqrt{2}} - \sqrt{1+\rho}\Phi^{-1}\left(1-\tfrac{\alpha}{2}\right) ,  \frac{\lambda \sigma_U-(1-\lambda)\sigma_L}{\lambda \sigma_U+(1-\lambda)\sigma_L} \times \frac{\beta}{\sqrt{2}} + \sqrt{1+\rho}\Phi^{-1}\left(1-\tfrac{\alpha}{2}\right) \right] \right)
\end{eqnarray*}
with the understanding that the first interval above is empty for small enough $x_2$.

Irrespective of the value taken by $x_2$, both intervals are centered at $\frac{\lambda \sigma_U-(1-\lambda)\sigma_L}{\lambda \sigma_U+(1-\lambda)\sigma_L} \times \frac{\beta}{\sqrt{2}}$, an expression that increases in $\lambda$ and takes value $0$ at $\lambda=\lambda^*$. The intervals' length does not depend on $\lambda$, and their union coincides with the larger of the two (whose identity depends on $x_2$). Again irrespective of the value of $x_2$, $X_1$ is distributed normally around $0$. By log-concavity of the Normal distribution (or by taking derivatives), the above probability therefore increases in $\lambda$ up to $\lambda^*$ and decreases in $\lambda$ thereafter conditionally on any $x_2$, hence also unconditionally. Furthermore, plugging in $\lambda \in \{0,1\}$ reveals symmetry about $0$: Switching $\lambda$ from $0$ to $1$ is equivalent to leaving $\lambda$ unchanged but replacing $X_1$ with $-X_1$. The probabilities of all intervals in the above display are , therefore, equally minimized at $\lambda \in \{0,1\}$ (although these minima correspond to different $\Delta$). This establishes that, if both of $(\Delta,\lambda)$ are concentrated out globally, one can restrict attention to one of $\lambda=0$ or $\lambda=1$.

We finally observe that the way in which $\sigma_L$ enters
\begin{eqnarray*}
E_{\Delta,1,c} = \left\{ Z^*_1 - \frac{\Delta}{\sigma_L}  \leq c \cap Z^*_2 \geq - c \right\}   \cup \left\{ Z^*_1+  Z^*_2 - \frac{\Delta}{\sigma_L}   \in \left[-\sqrt{2+2\rho} \Phi^{-1}\left(1-\tfrac{\alpha}{2}\right),\sqrt{2+2\rho} \Phi^{-1}\left(1-\tfrac{\alpha}{2}\right)\right] \right\}
\end{eqnarray*}
allows the simplification 
\begin{eqnarray*}
&& \inf_{\Delta \geq 0} \Pr(E_{\Delta,1,c}) \\
& =& \inf_{\Delta \geq 0}\Pr\left( \bigl\{ Z^*_1 - \Delta  \leq c \cap Z^*_2 \geq - c \bigr\}   \cup \left\{ Z^*_1+  Z^*_2 - \Delta  \in \left[-\sqrt{2+2\rho} \Phi^{-1}\left(1-\tfrac{\alpha}{2}\right),\sqrt{2+2\rho} \Phi^{-1}\left(1-\tfrac{\alpha}{2}\right)\right] \right\}\right).
\end{eqnarray*}

\paragraph{Step 3: For $\rho=0$, concentrating out $\Delta$.} 
For the remainder of this proof, suppose $\rho=0$. In view of step 2, also restrict attention to $\lambda=1$. This step's main claim is that $\Pr(E_{\Delta,1,c})$ is first increasing and then decreasing (possibly, although not in fact, all increasing or all decreasing) in $\Delta \geq 0$. Suppose the claim is true, then it follows that $\inf_{\Delta \in [0,\infty)}\Pr(E_{\Delta,1,c})$ is attained either at $\Delta=0$ or as $\Delta \to \infty$. In the former case, $\theta_U=\theta^*$, so that $CI_{MA}$ is obviously conservative. The latter limit is easily seen to equal $1-\alpha$, and this is indeed the (unattained) infimal coverage. 

It remains to show the main claim. Write $\gamma=\sqrt{2}\Phi^{-1}(1-\alpha/2)$, then (also using $\rho=0$) we have
\begin{equation*}
E_{\Delta,1,c} = \bigl\{ Z_1^* - \Delta \leq c \cap Z_2^* \geq - c \bigr\}  \cup \bigl\{ Z_1^*+Z_2^*  - \Delta  \in [-\gamma,\gamma] \bigr\},
\end{equation*}
where $(Z_1^*,Z_2^*)$ is bivariate standard Normal. We will henceforth think of $\Pr(E_{\Delta,1,c})$ as function of $\Delta$ with $(c,\gamma)$ fixed. Note that the condition on critical values translates as $2c \geq\gamma$. 

Using $\Phi(\cdot)$ and $\phi(\cdot)$ for the standard normal distribution and density functions, write
\begin{eqnarray*}
\Pr(E_{\Delta,1,c} | Z_2^*=z_2) = \begin{cases} \Phi(\gamma+\Delta-z_2)-\Phi(-\gamma+\Delta-z_2), & z_2 < -c \\ \Phi(\gamma+\Delta-z_2), & -c \leq z_2 \leq -c+\gamma \\ \Phi(\Delta+c), & z_2 > -c+\gamma \end{cases}
\end{eqnarray*}
and therefore (the last step below will be elaborated after the display)
\begin{eqnarray}
&& \frac{d\Pr(E_{\Delta,1,c})}{d\Delta} \notag \\
&=& \frac{d\int_{-\infty}^\infty \Pr(E_{\Delta,1,c} | Z_2^*=z_2)\phi(z_2)dz_2}{d\Delta} \notag \\
&=& \int_{-\infty}^{-c+\gamma} \phi(\gamma+\Delta-z_2) \phi(z_2) dz_2 - \int_{-\infty}^{-c} \phi(-\gamma+\Delta-z_2) \phi(z_2) dz_2 + \int_{-c+\gamma}^{\infty} \phi(\Delta+c) \phi(z_2) dz_2  \notag \\
&=& \underset{A}{\underbrace{\sqrt{2} \left(\phi\left(\frac{\gamma+\Delta}{\sqrt{2}}\right)-\phi\left(\frac{-\gamma+\Delta}{\sqrt{2}}\right)\right)\Phi\left(\frac{\gamma-\Delta-2c}{\sqrt{2}}\right)}} + \underset{B}{\underbrace{\phi(\Delta+c)\Phi(c-\gamma)}}. \label{eq:AB}
\end{eqnarray}
To see the last step, note first that $\int_{-c+\gamma}^{\infty} \phi(\Delta+c) \phi(z_2) dz_2$ simplifies to $B$. Next, $$(Z_1^*,Z_2^*)=(\gamma+\Delta-z_2,z_2)\Leftrightarrow (X_1,X_2)=\left(\frac{\gamma+\Delta}{\sqrt{2}},\frac{2z_2-\gamma-\Delta}{\sqrt{2}}\right),$$
where $(X_1,X_2)$ is as in \eqref{eq:reparameterize}. Because $\rho=0$ implies that $(X_1,X_2)$ is standard normal, we have
\begin{multline*}
\int_{-\infty}^{-c+\gamma} \phi(\gamma+\Delta-z_2) \phi(z_2) dz_2 = \int_{-\infty}^{-c+\gamma} \phi\left(\frac{\gamma+\Delta}{\sqrt{2}}\right) \phi\left(\frac{2z_2-\gamma-\Delta}{\sqrt{2}}\right) dz_2 \\
= \sqrt{2} \int_{-\infty}^{(\gamma-\Delta-2c)/\sqrt{2}} \phi\left(\frac{\gamma+\Delta}{\sqrt{2}}\right)\phi(t)dt = \sqrt{2} \phi\left(\frac{\gamma+\Delta}{\sqrt{2}}\right) \Phi\left(\frac{\gamma-\Delta-2c}{\sqrt{2}}\right).
\end{multline*}
A similar computation for $\int_{-\infty}^{-c} \phi(-\gamma+\Delta-z_2) \phi(z_2) dz_2$ and rearrangement of terms yield term $A$ in \eqref{eq:AB}.

Term $A$ equals zero at $\Delta=0$ and then becomes negative. Term $B$ is positive throughout. Because all terms vanish as $\Delta \to \infty$, it is not useful to directly take further derivatives. However, we can compare the terms' relative magnitude. In particular, we will see that $|A|/|B|$ increases in $\Delta$, hence $d\Pr(E_{\Delta,1,c})/d\Delta$ has at most one sign change and that sign change (if it occurs) is from positive to negative, establishing the claim.

To see monotonicity of $|A|/|B|$, write
\begin{eqnarray*}
\frac{|A|}{|B|} &=& \sqrt{2} \times \frac{\phi\left(\frac{-\gamma+\Delta}{\sqrt{2}}\right)-\phi\left(\frac{\gamma+\Delta}{\sqrt{2}}\right)}{\phi(\Delta+c)} \times \frac{\Phi\left(\frac{\gamma-\Delta-2c}{\sqrt{2}}\right)}{\Phi(c-\gamma)} \\
&=& \sqrt{2} \times \frac{\exp\left(-\frac{1}{4}\left(\gamma^2+\Delta^2-2\gamma \Delta\right)\right)-\exp\left(-\frac{1}{4}\left(\gamma^2+\Delta^2+2\gamma \Delta\right)\right)}{\exp\left(-\frac{1}{2}\left(\Delta^2+c^2+2\Delta c\right)\right)} \times \frac{\Phi\left(\frac{\gamma-\Delta-2c}{\sqrt{2}}\right)}{\Phi(c-\gamma)}  \\
&= &  \biggl(\exp\left(\tfrac{\Delta^2}{4}+\Delta c + \tfrac{\gamma \Delta}{2}\right)-\exp\left(\tfrac{\Delta^2}{4}+\Delta c - \tfrac{\gamma \Delta}{2}\right)\biggr) \Phi\left(\frac{\gamma-\Delta-2c}{\sqrt{2}}\right) \times \text{const.},
\end{eqnarray*}
where ``const." absorbs terms that do not depend on $\Delta$. The derivative of this expression with respect to $\Delta$ (and dropping the multiplicative constant) is
\begin{eqnarray*}
&& \underset{C}{\underbrace{\left(\frac{\Delta+2c+\gamma}{2}\exp\left(\tfrac{\Delta^2}{4}+\Delta c + \tfrac{\gamma \Delta}{2}\right)-\frac{\Delta+2c-\gamma}{2}\exp\left(\tfrac{\Delta^2}{4}+\Delta c - \tfrac{\gamma \Delta}{2}\right)\right)}} \Phi\left(\frac{\gamma-\Delta-2c}{\sqrt{2}}\right) \\
&&- \frac{1}{\sqrt{2}} \biggl(\exp\left(\tfrac{\Delta^2}{4}+\Delta c + \tfrac{\gamma \Delta}{2}\right)-\exp\left(\tfrac{\Delta^2}{4}+\Delta c - \tfrac{\gamma \Delta}{2}\right)\biggr) \phi\left(\frac{\gamma-\Delta-2c}{\sqrt{2}}\right) \\
&\geq & \left(\frac{\Delta+2c+\gamma}{2}\exp\left(\tfrac{\Delta^2}{4}+\Delta c + \tfrac{\gamma \Delta}{2}\right)-\frac{\Delta+2c-\gamma}{2}\exp\left(\tfrac{\Delta^2}{4}+\Delta c - \tfrac{\gamma \Delta}{2}\right)\right)  \frac{\frac{\Delta+2c-\gamma}{\sqrt{2}}}{\left(\frac{\Delta+2c-\gamma}{\sqrt{2}}\right)^2+1}\phi\left(\frac{\Delta+2c-\gamma}{\sqrt{2}}\right) \\
&&- \frac{1}{\sqrt{2}} \biggl(\exp\left(\tfrac{\Delta^2}{4}+\Delta c + \tfrac{\gamma \Delta}{2}\right)-\exp\left(\tfrac{\Delta^2}{4}+\Delta c - \tfrac{\gamma \Delta}{2}\right)\biggr) \phi\left(\frac{\gamma-\Delta-2c}{\sqrt{2}}\right),
\end{eqnarray*}
using that $C\geq 0$ and $\Phi(-t)\geq \frac{t}{t^2+1} \phi(t)$. In order to sign this, divide through by $\phi(\dots)$ (both are the same by symmetry of $\phi(\cdot)$) as well as by $\exp\left(\frac{\Delta^2}{4}+\Delta c - \frac{\gamma \Delta}{2}\right)$, multiply through by $\sqrt{2}$ as well as $\left(\frac{(\Delta+2c-\gamma)^2}{2}+1\right)$, and rearrange terms to conclude that the last expression above has the same sign as
\begin{eqnarray*}
&& \left(\frac{\Delta+2c+\gamma}{\sqrt{2}}\times\frac{\Delta+2c-\gamma}{\sqrt{2}} -  \left(\frac{(\Delta+2c-\gamma)^2}{2}+1\right)\right)\exp(\gamma \Delta) \\
&& ~~~~~~~~~~~~~~~~~~~~~~~~- \left(\frac{\Delta+2c-\gamma}{\sqrt{2}}\times\frac{\Delta+2c-\gamma}{\sqrt{2}}  -  \left(\frac{(\Delta+2c-\gamma)^2}{2}+1\right)\right)  \\
&=& \left(\frac{(\Delta+2c+\gamma)(\Delta+2c-\gamma)}{2} - \frac{(\Delta+2c-\gamma)^2}{2} - 1 \right)\exp(\gamma \Delta) + 1 \\
&=& \left(\gamma(\Delta+2c-\gamma) - 1 \right)\exp(\gamma \Delta) + 1 .
\end{eqnarray*}
At $\Delta=0$, this simplifies to $\gamma(2c-\gamma)$ and therefore is nonnegative if $2c\geq\gamma$. But one can also write
\begin{eqnarray*}
&& \frac{d}{d\Delta} \bigl(\left(\gamma(\Delta+2c-\gamma) - 1 \right)\exp(\gamma \Delta) + 1 \bigr) \\
&= & \gamma \exp(\gamma\Delta) + (\gamma(\Delta+2c-\gamma) - 1 )\gamma\exp(\gamma\Delta) \\
&= & \gamma^2(\Delta+2c-\gamma)\exp(\gamma\Delta),
\end{eqnarray*}
which is again nonnegative if $2c\geq\gamma$. Thus, $|A|/|B|$ is nondecreasing in $\Delta$ for all $\Delta\geq 0$, concluding the proof.  

\newpage

\bibliography{projection}
\newpage

\end{document}